\documentclass[reqno,a4paper]{amsart}
\usepackage{amssymb,latexsym,amsmath,enumerate}
\dedicatory{Dedicated to Raja --
dear friend, colleague, and inspirator.}
\newcommand{\norm}[1]{\parallel\!#1\!\parallel}
\newcommand{\twonorm}[1]{\parallel\!#1\!\parallel_2}

\numberwithin{equation}{section} 
\newtheorem{theorem}{Theorem}
\numberwithin{theorem}{section}
\newtheorem{lemma}{Lemma}
\numberwithin{lemma}{section}
\newtheorem{proposition}{Proposition}
\numberwithin{proposition}{section}

\theoremstyle{definition}
\newtheorem{definition}{Definition}

\theoremstyle{remark}

\newtheorem*{remark}{Remark}
\newtheorem*{convention}{Convention}
\newcommand{\z}{\mathbf{Z}}
\newcommand{\q}{\mathbf{Q}}
\newcommand{\rd}{\mathbf{R}^d}
\newcommand{\rr}{\mathbf{R}}
\newcommand{\cc}{\mathbf{C}}
\newcommand{\qp}{\mathbf{Q}_p}
\newcommand{\zp}{\mathbf{Z}_p}
\newcommand{\one}{\mathbf 1}
\newcommand{\ff}{{\mathbf F}}
\newcommand{\fo}{{\mathcal F}}
\newcommand{\foi}{{\mathcal F}^{-1}}
\newcommand{\gl}{\lambda}
\newcommand{\bn}{B_n}
\newcommand{\bmn}{B_{-n}}
\newcommand{\qmn}{q^{-n}}

\newcommand{\fq}{\mathbf{F}_q}
\newcommand{\cn}{\mathcal{C}_n}
\newcommand{\sn}{\mathcal{S}_n}
\newcommand{\dn}{\mathcal{D}_n}

\newcommand{\pa}{P^\alpha}

\newcommand{\qtt}{\mathbf{Q}_3[\sqrt{3}]}
\DeclareMathOperator{\supp}{supp}
\DeclareMathOperator{\ave}{ave}
\DeclareMathOperator{\kar}{char}
\DeclareMathOperator{\tr}{Tr}

\begin{document}
\title[Finite approximations]{FINITE APPROXIMATIONS OF PHYSICAL MODELS OVER
LOCAL FIELDS}
\author{Erik M. Bakken}\address{Department of Mathematical Sciences\\The Norwegian
University of Science and Technology\\7491 Trondheim\\Norway}
\email{erikmaki@math.ntnu.no}

\author{Trond  Digernes}
\address{Department of Mathematical Sciences\\The Norwegian
University of Science and Technology\\7491 Trondheim\\Norway}
\email{digernes@math.ntnu.no}

\thanks{The research of the second named author was partially supported by the Norwegian Research Council} 

\begin{abstract}
We show that the Schr\"odinger operator associated with a physical system over a local field can be approximated in a very strong sense by finite Schr\"odinger operators. Some striking numerical results are included at the end of the article.
\end{abstract}

\maketitle
\tableofcontents
\section{Introduction}

In \cite{DVV94} it was shown that a quantum mechanical Hamiltonian of the form $H=-\Delta + V$, acting in $L^2(\rd)$, with potential $V(x)\to\infty$ as $|x|\to\infty$, can be approximated in a very strong sense by finite quantum systems. In this note we present a similar theorem for quantum systems over a local field $K$.

The results of \cite{DVV94} were later extended to a setting of locally compact abelian groups in \cite{AGK00}. The results of the latter thus supersede both those of  \cite{DVV94} and of this article. However, the proofs of \cite{AGK00} used non-standard analysis. We have found it worthwhile to present a proof which does not rely on non-standard methods.

In \cite{DVV94} two proofs of the main convergence theorem were given: a functional analytic one and a probabilistic one. The latter gave a somewhat stronger convergence result for stochastic Hamiltonians. In the present note only functional analytic methods will be considered. A stochastic proof will be discussed in a forthcoming paper.

In an earlier article \cite{MR3127393} finite approximations over $\qp$ were treated. The current article supersedes that one; also, the proofs which were omitted there, are given here.

In Section~\ref{localfields} we give a quick review of local fields.
In Section~\ref{finapprox} we construct finite models for the Schr\"odinger operator over a local field, and  in Section~\ref{main} we prove the main convergence theorem.

In Section~\ref{numeric} we use our finite models to carry out a numerical investigation of the Schr\"odinger operator over the quadratic extension $\q_3[\sqrt{3}]$ of $\q_3$. We show that there is remarkable agreement between numerical and theoretical values for both eigenvalues and eigenfunctions. Both types of eigenfunctions (radial ones and those supported on single shells) appear already at the finite level.

\section{Local fields}\label{localfields}
We give here some quick facts about local fields. For a thorough treatment, see the classic treatise of A.~Weil \cite[Ch.~I]{MR0427267}; for a quicker review, see the book of Kochubei \cite[Ch.\ 1.3]{MR1848777}.

A local field is a non-discrete, locally compact field. The only connected local fields are $\rr$ and $\cc$. Disconnected local fields are, in fact, totally disconnected. 

Every local field comes equipped with a canonical absolute value which defines its topology. It is is induced by the Haar measure and is called \emph{module} in \cite{MR0427267}. It is Archimedean in the case of $\rr$ and $\cc$, and non-Archimedean in all other cases; it coincides with the usual absolute values for the fields $\rr$, $\cc$, and $\qp$. For a general local field $K$ we will denote the canonical absolute value by $|\cdot|$ (or by $|\cdot|_K$ if needed for clarity); for $\qp$ we will denote it by $|\cdot|_p$.
\begin{convention}
Since all local fields except $\rr$ and $\cc$ are (totally) disconnected, it is customary to reserve the term 'local field' for a (totally) disconnected, non-discrete, locally compact field. We will follow that convention here.
\end{convention}
With this convention, there are two main types of local fields:\\
\emph{Characteristic zero.}
The basic example of a local field of characteristic zero is the $p$-adic field $\qp$ ($p$ a prime number). Every local field of characterisitic zero is a finite extension of $\qp$ for some $p$.\\
\emph{Positive characteristic.}
Every local field of positive characteristic $p$ is isomorphic to the field $\mathbf F_q((t))$ of Laurent series over a finite field $\mathbf F_q$, where $q=p^f$ for some positive integer $f\geq1$. 

Let $K$ be a local field with canonical absolute value $|\cdot|$. Following standard notation, we set
\[
O = \{x \in K : |x|\leq 1\},\quad P = \{x \in K : |x| < 1\},\quad U = O\setminus P.
\]
$O$ is a compact subring of $K$, called the \emph{ring of integers}. It is a discrete valuation ring, i.e., a principal ideal domain with a unique  maximal ideal. $P$ is the unique non-zero maximal ideal of $O$, called the \emph{prime ideal}, and any element $\beta\in P$ such that $P=\beta O$ is called a \emph{uniformizer} (or a \emph{prime element}) of $K$. For $\qp$ one can choose $\beta=p$, and for $\mathbf F_q((t))$ one can take $\beta=t$.\\
The set $U$ coincides with the \emph{group of units} of $O$. The quotient ring $O/P$ is a finite field. If $q=p^f$ is the number of elements in $O/P$ ($p$: a prime number, $f$: a natural number) and $\beta$ is a uniformizer, then $|\beta|=1/q$, and the range of values of $|\cdot|$ is $q^N$, $N\in\mathbf Z$. Further, if $S$ is a complete set of representatives for the residue classes in $O/P$, every non-zero element $x\in K$ can be written uniquely in the form:
\[
x=\beta^{-m}(x_0+x_1\beta+x_2\beta^2+\cdots),
\]
where $m\in\mathbf Z$, $x_j\in S$, $x_0\not\in P$. With $x$ written in this form, we have $|x|=q^m$.

For a general field extension $K/F$ we use the following standard notation: $f=\text{index of inertia}$, and $e=\text{ramification index}$. These are connected through the formula $[K:F]=ef$. If $e=1$, the extension is unramified, and if $f=1$, the extension is totally ramified.

\subsection{Characters and Fourier transform}\label{ft}
We first fix a Haar measure $\mu$ on $K$, normalized such that $\mu(O)=1$. The Fourier transform $\fo$ on $K$ is given by 
\[(\fo f)(\xi)=\int_K f(x)\chi(-x\xi)\,dx\,,\]
where $\chi$ is a suitably chosen non-trivial character on $K$, and $dx$ refers to the Haar measure just introduced. For our set-up it will be essential to use a character of rank zero\footnote{We remind the reader that the rank of a character $\chi$ is defined as the largest integer $r$ such that $\chi\vert_{B_r}\equiv 1$.}. We describe a procedure for achieving this in the two main cases:
\subsubsection{Case 1: $\kar K=0$}\label{char zero} In this case $K$ is a finite extension of $\qp$, and a character of rank zero is obtained by setting 
\[\chi(x)=\chi_p\left(\tr_{K/\qp}(\beta^{-d}x)\right),\quad x\in K\,,\] 
where
\begin{itemize}
\item $\chi_p$ is the canonical character on $\qp$---i.e., $\chi_p(x)=\exp(2\pi i\{x\})$, $\{x\}=\text{fractional part of $x$}$.
\item $\tr_{K/\qp}:K\to\qp$ is the trace function associated with the extension $K/\qp$.
\item $\beta$ is a uniformizer as defined above.
\item $d$ is the \emph{exponent of the different} of the extension $K/\qp$. It is the  largest integer $d$ such that $\tr_{K/\qp}(x)\in\zp$ for all $x$ with 
$|x|\leq q^d$ (note that $d\geq0$ since $\tr_{K/\qp}:O\to\zp$).
\end{itemize}
\subsubsection{Case 2: $\kar K>0$}\label{char nonzero} In this case we may identify $K$ with the field 
$\mathbf{F}_q((t))$ of Laurent series in the indeterminate $t$ with coefficients from the finite field $\mathbf{F}_q$, $q=p^f$, consisting of elements of the form $x=\sum_{i=m}^\infty x_i t^i$, $x_i\in \mathbf F_q$, $m\in\z$. Let $\eta$ denote the canonical character on $\ff_q$, i.e., $\eta(x)=\exp\left(\frac{2\pi i}{p}\tr_{\ff_q/\ff_p}(x)\right)$, and define 
\[\chi(x)=\eta(x_{-1})\,,\]
where $x_{-1}$ refers to the expansion $x=\sum_{i=m}^\infty x_it^i$. Then $\chi$ is a rank zero character on $K=\fq((t))$.

Notice that any Fourier transform based on a rank zero character is an $L^2$-isometry with respect to the normalized Haar measure defined above (since $\fo\one_O=\one_O$ for any such Fourier transform $\fo$; here and elsewhere $\one$ denotes characteristic function). Thus $\foi=\fo^*$ is given by
\[(\foi f)(x)=(\fo^* f)(x)=\int_K f(y)\chi(xy)\,dy.\]
\begin{convention}\label{ftconv}For the rest of this article $\fo$ will denote a Fourier transform based on a rank zero character on $K$.
\end{convention}

\section{Finite approximations over a local field}\label{finapprox}
Our object of  study is a version of the Schr\"odinger operator, defined for $\qp$ in the book of Vladimirov, Volovich, Zelenov \cite{VVZ94}, and generalized to an arbitrary local field $K$ by Kochubei in \cite{MR1848777}:
\[H=D^\alpha+V\,,\]
regarded as an operator in $L^2(K)$. 
Here $\alpha>0$ \footnote{For a direct analog of the Laplacian one should set $\alpha=2$. However, as is customary in the non-Archimedean setting, one works with an arbitrary $\alpha>0$, since the qualitative behavior of the operator $H$ does not change with $\alpha>0$.}, $D=\foi Q\fo$ where $(Qf)(x)=|x| f(x)$ is the position operator, and $\fo$ is the Fourier transform on $L^2(K)$. $V$ (the potential) is multiplication by a \emph{radial} function: $(Vf)(x)=v(x)f(x)$, $v(x)=w(|x|)$ for some function $w$ defined on $[0,\infty)$. We assume $v$ to be non-negative and continuous and that $v(x)\to\infty$ as $|x|\to\infty$. 

Due to a conflict of notation later in this article, we will use the symbol $P$ for the differentiation operator (instead of $D$). With this notation we have
\[H=P^\alpha+V\,.\] 
The operator $H$ has been thoroughly analyzed (see \cite{VVZ94} for $K=\qp$ and \cite{MR1848777} for general $K$): It is self-adjoint on the domain $\{f\in L^2(K):P^\alpha f+Vf\in L^2(K)\}$, has discrete spectrum, and all eigenvalues have finite multiplicity. Our next task is to set up a finite model for this operator.
\subsection{Finite model}\label{gn}
Keep the above notation, i.e.: $K$ is a local field, $q=p^f$ is the number of elements in the finite field $O/P$, $\beta$ is a uniformizer, and $S$ is a complete set of representatives for $O/P$. For each integer $n$ set $\bn=\beta^{-n}O=\text{ball of radius $q^n$}$. Then $\bn$ is an open, additive subgroup of $K$. For $n>0$ we set $G_n=\bn/\bmn$. Then $G_n$ is a finite group with $q^{2n}$ elements. Since the subgroup $\bmn$ will appear quite frequently, we will often denote it by $H_n$, to emphasize its role as a subgroup. So $H_n=\bmn=\beta^nO=\text{ball of radius $q^{-n}$}$, and $G_n=H_{-n}/H_n$. Each element of $G_n$ has a unique representative of the form $a_{-n}\beta^{-n}+a_{-n+1}\beta^{-n+1}+\dots+a_{-1}\beta^{-1}+ a_0+a_1\beta+\dots+ a_{n-2}\beta^{n-2}+a_{n-1}\beta^{n-1}$, $a_i\in S $. We denote this set by $X_n$, and call it \emph{the canonical set of representatives} for $G_n$; we also give it the group structure coming from its natural identification with $G_n$.

Let again $\mu$ denote the normalized Haar measure on $K$ (cfr.\ \ref{ft}). Since $H_n$ is an open subgroup of $K$, we obtain a Haar measure $\mu_n$ on  $G_n=H_{-n}/H_n$  
by setting $\mu_n(x+H_n)=\mu(x+H_n)=\mu(H_n)=q^{-n}$, for $x+H_n\in G_n$.

So each "point" $x+H_n$ of $G_n$ has mass $q^{-n}$, and the total mass of $G_n$ is $q^{2n}\cdot q^{-n}=q^n$.

With this choice of Haar measure on $G_n$ the mapping which sends the characteristic function of the point $x+H_n$ in $G_n$ to the characteristic function of the subset $x+H_n$ of $K$, is an isometric imbedding of $L^2(G_n)$ into $L^2(K)$. We regard operators on $L^2(G_n)$ as operators on $L^2(K)$ via this imbedding, by setting them equal to 0 on the orthogonal complement of the image of $L^2(G_n)$ in $L^2(K)$.

We introduce the following subspaces of $L^2(K)$, along with their orthogonal projections :
\begin{itemize}
\item $\mathcal C_n=\{f\in L^2(K)|\supp(f)\subset B_{n}\}.$ The corresponding orthogonal projection is denoted by $C_n$ and is given by: $C_nf=\one_{B_n}f$.
\item $\mathcal S_n=\{f\in L^2(K)|\text{$f$ is locally constant of index $\leq q^{-n}$}\}.$ The corresponding orthogonal projection is denoted by $S_n$ and is given by:\\ $(S_nf)(x)=q^n\int_{H_n}f(x+y)\,dy
=\frac{1}{\mu(H_n)}\int_{H_n}f(x+y)\,dy=\ave(f,n,x)$, where we have introduced the notation $\ave(f,n,x)$ for the average value of $f$ over $x+H_n$.
\item $\mathcal D_n=\mathcal C_n\cap\mathcal S_{n}.$ 
The corresponding orthogonal projection is denoted by $D_n$.
\end{itemize}
Note that $L^2(G_n)$ is mapped onto $\mathcal D_n$ via the isometric imbedding mentioned above. Thus $L^2(G_n)$ can be thought of as the set of functions on $K$ which have support in $B_n$ and which are invariant under translation by elements of $H_n\,(=B_{-n})$.
\begin{lemma}
The projections $C_n$ and $S_n$ commute, thus the projection $D_n$ onto the subspace $\mathcal D_n$ is given by:
\[D_n=C_nS_n=S_nC_n.\]
\end{lemma}
\begin{proof}
\begin{align*}
(S_{n}C_nf)(x)&=q^n\int_{H_n}(C_nf)(x+y)\,dy=q^n\int_{H_n}\one_{B_n}(x+y)f(x+y)\,dy\\
&\stackrel{(*)}{=}q^n\int_{H_n}\one_{B_n}(x)f(x+y)\,dy=\one_{B_n}(x)q^n\int_{H_n}f(x+y)\,dy\\&=(C_nS_{n}f)(x)\,
\end{align*}
where the equality $(*)$ follows from ultrametricity, namely: $x+y\in B_n\Longleftrightarrow x\in B_n$ when $y\in H_n=B_{-n}$.
\end{proof}

We next show that the Fourier transform behaves nicely with respect to these subspaces.
\begin{proposition}\label{propcommute}
We have:
\[\fo\mathcal C_n=\mathcal S_{n},\quad \fo\mathcal S_n=\mathcal C_{n},\text{ and hence }\fo\mathcal D_n=\mathcal D_n,
\]
and the same relations hold with $\foi$ in place of $\fo$.\\
As a consequence, the following commutation relations hold:
\[\fo C_n=S_n\fo,\quad \fo S_n=C_n\fo,\quad \fo D_n=D_n\fo.\]
\end{proposition}
\begin{proof}
Let $f\in\mathcal{C}_n$ and take any $h\in H_n$. Then 
\begin{align*}
(\mathcal F f)(\xi + h)&=\int_K f(x)\chi(-x(\xi +h))dx=\int_{B_n} f(x)\chi(-x\xi)\chi(-xh)dx\\
&=\int_{B_n} f(x)\chi(-x\xi)dx=\int_K f(x)\chi(-x\xi)dx\\
&=(\mathcal F f)(\xi)
\end{align*}
since $|xh|\leq q^n\cdot q^{-n}=1$ and $\chi$ has rank zero. This proves $\fo\cn\subset\sn$.\\
Next let $f\in\mathcal S_n$ and assume $(\mathcal Ff)(\xi)\neq0$, $|\xi|=q^m$. We must show that $m\leq n$. For any $h$ with $|h|\leq q^{-n}$ we have
\begin{align*}
(\mathcal F f)(\xi)&=\int_K f(x)\chi(-x\xi)dx=\int_K f(x+h)\chi(-x\xi)dx\\
&=\int_{K} f(x)\chi(-(x-h)\xi)dx=\int_K f(x)\chi(-x\xi)\chi(h\xi)dx\\
&=(\mathcal F f)(\xi)\chi(h\xi)\\
\intertext{which, since $(\mathcal Ff)(\xi)\neq0$, gives}
\chi(h\xi)&=1 \text{ for all $h\in B_{-n}\,.$}
\end{align*}
This means that $\chi$ is identically equal to $1$ on the ball $\xi\cdot B_{-n}=B_{-n+m}$, and since $\chi$ has rank zero, we must have $-n+m\leq0$, i.e., $m\leq n$. This proves $\fo\sn\subset\cn$. Since obviously the same relations hold with $\foi$ in place of $\fo$, we have equalities everywhere, i.e., $\fo\mathcal C_n=\mathcal S_{n}$, $\fo\mathcal S_n=\mathcal C_{n}$, and hence $\fo\mathcal D_n=\mathcal D_n$.

As for the commutation relations: 
The relations just proved\ -- and the same ones with  $\foi=\fo^*$ instead of  $\fo$ -- imply that $S_n\fo C_n=\fo C_n$, $C_n\fo S_n=\fo S_n$, $S_n\fo^*C_n=\fo^*C_n$, $C_n\fo^*S_n=\fo^*S_n$. Taking adjoints and combining, we get $\fo C_n=S_n\fo$ and $\fo S_n=C_n\fo$. Multiplying  $\fo C_n=S_n\fo$ by $C_n$ on the left and multiplying $\fo S_n=C_n\fo$ by $C_n$ on the right gives $C_n\fo C_n=C_nS_n\fo=\fo S_nC_n$, i.e.,  
$\fo D_n=D_n\fo.$ 
\end{proof}


\subsection{Fourier transform at the finite level}
We need to establish a relation between the Fourier transforms on $K$ and $G_n$. 

So let as before $\chi$ be a rank zero character on $K$ and let $\fo$ be the associated Fourier transform. Like any additive character on a field, $\chi$ gives rise to a symmetric bi-character $\mathcal X$
on $K$ by setting $\mathcal X(x,y)=\chi(xy)$. It descends to a bi-character on $G_n=B_n/B_{-n}$, since if $x'=x+h$, $y'=y+k$ with $h,k\in B_{-n}$, then 
$\mathcal X(x+h,y+k)=\chi((x+h)(y+k))=\chi(xy)\chi(xk)\chi(hy)\chi(hk)=\chi(xy)=\mathcal X(x,y)$ (the arguments in the last three factors of the product all have absolute value $\leq1$). So we can define a bi-character $\mathcal X_n$ on $G_n$ by setting $\mathcal{X}_n([x],[y])=\chi(xy)$. Since $\mathcal X$ is non-degenerate on $K$, so is $\mathcal X_n$ on $G_n$. Indeed, if $x\in B_n$, $|x|=q^m$, and 
$\mathcal{X}_n([x],[y])=\chi(xy)=1$ for all $y\in B_n$, then $\chi=1$ on the ball $x\cdot B_n=B_{n+m}$, which implies $B_{n+m}\subset B_0=O$ since $\chi$ has rank $0$. But this means that $m+n\leq0$, i.e., $m\leq -n$, and so $x\in B_{-n}$, i.e., $x=0$ as an element of $G_n=B_n/B_{-n}$.\\    Setting $\chi_{n,[y]}([x])=\mathcal{X}_n([x],[y])$, it follows that the characters $\chi_{n,[y]}$ exhaust all of  
$\hat{G_n}$ as $[y]$ runs through $G_n$, i.e., the bi-character $\mathcal X_n$ implements the self-duality of the finite abelian group $G_n$. The canonical choice for an $L^2$-isometric Fourier transform on $G_n$ is then given by (recall that $G_n$ has $q^{2n}$ elements):
\begin{align}\label{ftn}
\begin{split}
(\mathcal{F}_nf)([x])&=\frac{1}{\sqrt{|G_n|}}\sum_{[y]\in G_n}f([y])\mathcal{X}_n(-[x],[y])\\&=\qmn\sum_{[y]\in G_n}f([y])\chi(-xy),\quad[x]\in G_n,\quad f\in L^2(G_n)\,,
\end{split}
\intertext{or, in terms of the set of representatives $X_n$,}
\begin{split}
(\mathcal{F}_nf)(x)&=\qmn\sum_{y\in X_n}f(y)\mathcal{X}_n(-x,y)\\
&=\qmn\sum_{y\in X_n}f(y)\chi(-xy),\quad x\in X_n,\quad f\in L^2( X_n)\,.
\end{split}
\end{align}
The following result is now more or less obvious, but we state it as a proposition because of its importance. It plays a crucial role in the proof of the main convergence theorems, and simplifies matters considerably compared to the situation over $\rr$, where the relation between the finite and infinite Fourier transform was much more complicated (see \cite[p.\ 626--627]{DVV94}).
\begin{proposition} Let the Fourier transforms $\fo$ and $\fo_n$ be as above.
Then $\fo$ leaves the space $\dn$ $\left(\simeq L^2(G_n)\right)$ invariant, and
\begin{equation}
\fo\vert_{\dn}=\fo_n, \text{ i.e., } \fo_n=\fo D_n=D_n\fo\,.
\end{equation}
\end{proposition}
\begin{proof}
The first part of the statement has already been proved (Proposition~\ref{propcommute}). For the second part, take any $f\in\dn$ and let $x\in\bn$. Then:
\begin{align}
(\fo f)(x)&=\int_K f(y)\chi(-xy)\,dy\stackrel{f\in\cn}{=}\int_{\bn} f(y)\chi(-xy)\,dy\\
&=\sum_{z\in X_n}\int_{z+H_n}f(y)\chi(-xy)\,dy\stackrel{(*)}{=}\sum_{z\in X_n}f(z)\chi(-xz)\qmn\\
&=(\fo_nf)(x)\,,
\end{align}
where the equality $(*)$ follows from the fact that the function $y\to f(y)\chi(-xy)$ is constant on $z+H_n$ (since $x\in\bn$) and $\mu(z+H_n)=\qmn$. 
\end{proof}
 
\subsection{Dynamical operators at the finite level}

For the finite versions of the dynamic operators we take their compressions by $D_n$, i.e., $V_n=D_nVD_n$, $Q_n=D_nQD_n$, $P_n=D_nPD_n=D_n\foi Q\fo D_n=\foi Q_n\fo=\foi_n Q_n\fo_n$. Before computing what these operators do to an $f\in L^2(G_n)$, let us find out what the projection $S_n$ does to a radial function $v(x)=w(|x|)$:
\begin{align*}
(S_nv)(x)&=\frac{1}{\mu(H_n)}\int_{H_n}v(x+h)\,dh=\frac{1}{\mu(H_n)}\int_{H_n}w(|x+h|)\,dh\\
&=\ave(v,n,x)\stackrel{(*)}{=}\begin{cases}v(x),&|x|>\qmn\\
\ave(v,n,0),&|x|\leq \qmn\,,
\end{cases}
\end{align*}
where again $\ave(v,n,x)$ means the average value of $v$ over $x+H_n$, and where ultrametricity was used in the equality $(*)$. 

Next we compute the effect of the finite operators on an $f\in L^2(G_n)$. For $V_n$ we get, remembering that $V$ is multiplication by a radial function $v$:
\begin{align*}
(V_nf)(x)&=(D_nVD_nf)(x)=(C_nS_nVf)(x)=\one_{B_n}(x)(S_nVf)(x)\\&=\one_{B_n}(x)\frac{1}{\mu(H_n)}\int_{H_n}(Vf)(x+h)\,dh\\&=\one_{B_n}(x)\frac{1}{\mu(H_n)}\int_{H_n}v(x+h) f(x+h)\,dh\\&=\one_{B_n}(x)\frac{1}{\mu(H_n)}[\int_{H_n}v(x+h)\,dh]f(x)\\&=\ave(v,n,x)f(x)=\begin{cases}v(x)f(x),&|x|>\qmn\\
\ave(v,n,0)f(0),&|x|\leq\qmn
\end{cases}\\&=
\begin{cases}(Vf)(x),&|x|>\qmn\\
\ave(v,n,0)f(0),&|x|\leq\qmn
\end{cases}\quad\left(f\in L^2(G_n)\right)\,.
\end{align*}
In particular, for the operator $Q_n$ this gives, writing $q(x)=|x|$:
\begin{align*}(Q_nf)(x)&=\ave(q,n,x)f(x)=\begin{cases}|x|f(x),&|x|>\qmn\\\ave(q,n,0)f(0),&|x|\leq\qmn\end{cases}\\
&=\begin{cases}(Qf)(x),&|x|>\qmn\\\ave(q,n,0)f(0),&|x|\leq\qmn\end{cases}\quad\left(f\in L^2(G_n)\right)\,.
\end{align*}
For $P_n$ we get
\begin{align*}
P_nf&=D_nPD_nf=C_nS_nPf=C_nS_n\foi Q \fo f=C_n\foi C_n Q \fo f=C_n\foi Q \fo f
\\&=C_n Pf\quad\left(f\in L^2(G_n)\right).
\end{align*}
We now set $H_n=P_n^\alpha+V_n$, the Hamiltonian for the finite model, and aim to show that the analog of Theorem~4 in \cite{DVV94} holds in the present setting.

\section{Convergence of the finite models}\label{main} Keep the notation and assumptions of the previous section. 

There are two main steps to proving the analog of Theorem~4 of \cite{DVV94}: Establishing the convergence $H_n\to H$ in the strong resolvent sense, and proving a form of uniform compactness for the resolvents $(I+H_n)^{-1}$. The proofs follow a pattern similar to that of \cite{DVV94}, but we are able to simplify some of the arguments, partly due to the non-Archimedean nature of $K$.

As for strong convergence of the resolvents: According to \cite{MR1848777}, Section~3.2, the space $\mathcal D$ of locally constant functions with compact support is a core for the Hamiltonian $H=P^\alpha+V$. Hence it is a common core for all the Hamiltonians $H_n$ ($n\ge 1$) and $H$. For $f\in\mathcal D$ we have $f\in \dn$ for large $n$, hence $\lim_n Q_n^\alpha f=\lim_n D_nQ^\alpha f=Q^\alpha f$ in the strong operator topology; further: $P_n^\alpha f=\foi_n Q_n^\alpha\fo_n f=\foi Q_n^\alpha\fo f\to  \foi Q^\alpha\fo f=P^\alpha f$, \footnote{Proving the limit $P_n^\alpha f\to P^\alpha f$ required considerable effort in \cite{DVV94}, due to the fact that the Fourier transforms at the finite and infinite level did not match up nicely. Here the finite Fourier transform is simply the restriction of the infinite one, and the limit becomes a triviality.} and $V_nf\to Vf$. Thus $H_nf=(P_n^\alpha+V_n)f\to (P^\alpha+V)f=H f$ for all $f\in\mathcal D$. Here we have used the obvious fact that $C_n\to I$, and hence $S_n=\fo C_n\foi\to I$ and $D_n=C_nS_n=S_nC_n\to I$, in the strong operator topology. 

By Theorem VIII.25 of \cite{MR751959} it now follows that $H_n\to H$ in the strong resolvent sense.

The compactness of the resolvent $(I+H)^{-1}$ follows by classical arguments (see, e.g., \cite[p.~623]{DVV94} for the case $L^2(\rd)$; the same proof works for $L^2(K)$).  For the resolvents $(I+H_n)^{-1}$ we need a form of uniform compactness which is formulated as follows: 
\begin{definition}[Uniform compactness]\label{uc}
A sequence of bounded operators $(M_n)$ on a Hilbert space $\mathcal{H}$ is said to satisfy a condition of \emph{uniform compactness} if the following conditions hold:
\begin{enumerate}
\item The sequence $(M_n)$ is uniformly bounded.
\item There are subspaces $L_n$ with $L_n$ invariant under $M_n$ such that for every sequence $(g_n)$ with $g_n\in L_n$ and $||{g_n}||\leq1$, the sequence 
$(M_ng_n)$ is relatively compact in $\mathcal H$.
\end{enumerate}
\end{definition}
\begin{remark}
Notice that the individual operators $M_n$ are not required to be compact on $\mathcal H$ (and in our applications they will not be). Still, if the above conditions are fulfilled, we will say that the sequence $(M_n)$ is uniformly compact, even if the individual $M_n$ are not compact.
\end{remark}
For our purposes the usefulness of uniform compactness lies in the following two results. They give a strong connection between the spectral data of the operators in an approximating sequence $(M_n)$ and their strong limit $M$. 
\begin{lemma}\label{acc}
Let $M_n$, $L_n$ be as in Definition~\ref{uc}, and assume that the sequence $M_n$ converges strongly to a bounded operator $M$. Assume further that there are eigenvectors $g_n$ and corresponding eigenvalues $\gl_n$ such that $g_n\in L_n$, $||g_n||=1$ and $M_n g_n=\gl_n g_n$. Then any non-zero cluster point $\gl_0$ of the sequence $(\gl_n)$ is an eigenvalue of $M$, and there is a subsequence of $(g_n)$ which converges to a vector $g$ such that $Mg=\gl_0 g$. 
\end{lemma} 
\begin{proof}
By uniform boundedness, all the $\gl_n$ are confined to a bounded set. Hence there is a subsequence of $(\gl_n)$ (still written $(\gl_n)$ after re-indexing) which converges to a scalar $\gl_0$, say, with $\gl_0\neq0$. By uniform compactness, $M_ng_n$ has a convergent subsequence (again written $M_ng_n$ after re-indexing). It follows that the sequence $g_n=\frac{1}{\gl_n}M_ng_n$ converges to an element $g$, say. Since $M_n\to M$ strongly, it follows that $M_ng_n\to Mg$;  indeed, remembering that the $M_n$ are uniformly bounded: $\norm{M_ng_n-Mg}\leq\norm{M_n}\,\norm{g_n-g}+\norm{M_ng-Mg}\to 0$. So altogether we have: $g=\lim g_n=\lim \frac{1}{\gl_n}M_ng_n=\frac{1}{\gl_0}Mg$, i.e., $Mg=\gl_0g$.
\end{proof}
\emph{Notation:} We let $\sigma_p(A)$ denote the set of \emph{positive} eigenvalues of an operator $A$. Further, for a self-adjoint $A$ we let $P^A$ denote the projection valued measure of $A$, and for a projection $E$ we let $r(E)$ denote its range. 
\begin{proposition}[Cfr.\ Lemma 3 in \cite{DVV94}]\label{mainprop}
Keep the notation and assumptions of the previous lemma. In addition, assume the following: (i) The operators $M_n,M$ are self-adjoint, and $0\leq M,M_n\leq I$, (ii) $M$ is compact on $\mathcal H$, and $M_n$ is compact on $L_n$. Then the following hold:
\begin{enumerate}
\item If $J$ is a compact subset of $(0,1]$ with $J\cap\sigma_p(M)=\emptyset$, then $J\cap\sigma_p(M_n)=\emptyset$ for large $n$.
\item If $\gl\in\sigma_p(M)$, there exists a sequence $(\gl_n)$ with $\gl_n\in\sigma(M_n)$ such that $\gl_n\to\gl$. Further, if $J$ is a compact neighborhood of an eigenvalue $\gl\in\sigma_p(M)$, not containing any other eigenvalues of $M$, then any sequence $(\gl_n)$ with $\gl_n\in\sigma_p(M_n)\cap J$ converges to $\gl$.
\item Let $\gl$ and $J$ be as in (2). Then $\dim P^{M_n}(J) = \dim P^M(J)$ for large $n$, and for each orthonormal basis $\{e_1,\dots,e_m\}$ for $r\left(P^{M}(J)\right)$ there is, for each $n$, an orthonormal basis $\{e_1^n,\dots,e^n_m\}$ for $r\left(P^{M_n}(J)\right)$ such that $\lim_{n\to\infty}e^n_i=e_i$, $i=1,\dots,m$.  
\end{enumerate}
\end{proposition} 
\begin{proof}
(1)  If $J\cap\sigma_p(M_n)\neq\emptyset$ for arbitrarily large $n$, there are infinitely many $\gl_n$ in $J$. The sequence $(\gl_n)$ thus has a cluster point in $J$, and hence, by the previous lemma, $M$ has an eigenvalue in $J$.

(2) The first part follows from the fact that $M_n\to M$ strongly \cite[Thm.\ VIII.24, Vol.\ 1]{MR751959}. Now let $(\gl_{n_k})$ be all the eigenvalues of the various $M_n$ which lie in $J$, indexed in an arbitrary fashion. Then $(\gl_{n_k})$ has a cluster point in $J$, which by the previous lemma is an eigenvalue of $M$. Since $M$ has exactly one eigenvalue in $J$, it follows that the sequence $(\gl_{n_k})$ has exactly one cluster point in $J$, i.e., $(\gl_{n_k})$ converges to $\gl$.

(3) For ease of notation set $E_n=P^{M_n}(J)$ and $E=P^{M}(J)$. We first prove that $\dim E_n \leq \dim E$ for large $n$. Assume otherwise, and set $m=\dim E$. Then there exists a subsequence $E_{n_k}$ of $E_n$ such that $\dim E_{n_k}>m$ for all $k$. For each $k$, choose $m+1$ orthonormal eigenvectors $e^k_1,\dots,e^k_{m+1}$ for $r\left(E_{n_k}\right)$. By uniform compactness there is a subsequence of $(e^k_1)$ which converges to an eigenvector for $M$. Repeating the process for each of the remaining eigenvectors, we obtain a set of $m+1$ orthonormal eigenvectors for $M$, a contradiction. This proves $\dim E_n \leq \dim E$ for large $n$. The converse inequality $\dim E_n \geq \dim E$ follows from \cite[Thm.\ VIII.24, Vol.\ 1]{MR751959}: Since $M_n\to M$ strongly, then $E_n\to E$ strongly. For finite dimensional projections this implies $\dim E_n\geq\dim E$ for large $n$.

For the last statement take any orthonormal basis $\{e_1,\dots,e_m\}$ for $r(E)$. Let us first show  that the set $\{E_ne_1,\dots,E_ne_m\}$ is linearly independent for large 
$n$. Assume to the contrary that it is linearly dependent for arbitrarily large $n$, and let $1>\epsilon>0$ be given. By strong convergence there is an $n_0$ such that $\norm{E_n e_j-e_j}<\epsilon$ for $n\geq n_0$, $j=1\dots m$. Pick an $n>n_0$ such that the set $\{E_ne_1,\dots,E_ne_m\}$ is linearly dependent. From a linear dependence relation for this set, pick the term with the largest coefficient -- $E_ne_i$, say -- and solve for it. Then we have
\[E_ne_i=\sum_{j=1,j\neq i}^m\alpha_j E_n e_j\]
with $|\alpha_j|\leq 1$. Take the inner product with $e_i$ on both sides to get
\begin{align}
\langle E_ne_i,e_i\rangle&=\sum_{j=1,j\neq i}^m\alpha_j 
\langle E_n e_j,e_i\rangle\,,\\
\intertext{which gives}
\langle E_ne_i-e_i,e_i\rangle+1&=\sum_{j=1,j\neq i}^m\alpha_j \langle E_n e_j-e_j,e_i
\rangle.
\end{align}
For the left hand side we have $|\langle E_ne_i-e_i,e_i\rangle+1|\geq 1-\epsilon$, and for the right hand side: $|\sum_{j=1,j\neq i}^m\alpha_j \langle E_n e_j-e_j,e_i\rangle|\leq (m-1)\epsilon$. For $\epsilon<1/m$ this gives a contradiction. Hence the set $\{E_ne_1,\dots,E_ne_m\}$ is linearly independent for large $n$.
Now perform a Gram-Schmidt orthonormalization on this set to obtain an orthonormal basis $\{e_1^n,\dots,e_m^n\}$ for $r(E_n)$. An elementary, but somewhat tedious, calculation then shows that $\lim_{n\to\infty}e^n_i=e_i$, $i=1,\dots,m$. 
\end{proof}

We are now ready to prove a key result, namely that the sequence $(I+H_n)^{-1}$ is uniformly compact in the sense of Definition~\ref{uc} (see Proposition~\ref{propuc}). This will pave the way for establishing our main result (Theorem~\ref{mainthm}). To prove uniform compactness we will use the following version of the Kolmogorov-Riesz compactness criterion; it is proved for the case $L^2(\rd)$ in \cite[Corollary 7]{MR2734454}, and the same proof works for $L^2(K)$:
\begin{proposition}\label{KR}
Let $F$ be a subset of $L^2(K)$. Then $F$ is relatively compact if the following conditions are fulfilled:
\begin{enumerate}
\item $\sup_{f\in F}\twonorm{f}\leq C$ for some positive constant $C$.
\item $\lim_{r\to\infty}\sup_{f\in F}\int_{|x|\geq r}|f(x)|^2\,dx=0$.
\item $\lim_{\rho\to\infty}\sup_{f\in F}\int_{|\xi|\geq \rho}|\mathcal F f(\xi)|^2\,d\xi=0$.
\end{enumerate}
\end{proposition}

\begin{proposition}\label{propuc}
With $M_n=(I+H_n)^{-1}$, $L_n=\dn\simeq L^2(G_n)$, and $\mathcal{H}=L^2(K)$, the resolvents 
$(I+H_n)^{-1}$ are uniformly compact in the sense of Definition~\ref{uc}.
\end{proposition}
\begin{proof}
Let $(g_n)$ be as in Definition~\ref{uc} and set $f_n=(1+H_n)^{-1}g_n$. Then
\begin{equation*}\label{normestimate_f}\twonorm{f_n}^2+\langle H_nf_n,f_n\rangle=\langle(I+H_n)f_n,f_n\rangle=\langle g_n,f_n\rangle\leq\twonorm{g_n}\,\twonorm{f_n}\,,
\end{equation*}
and so $\twonorm{f_n}\leq1$ since $\langle H_nf_n,f_n\rangle\geq0$, and it also follows that $\twonorm{V_n^{1/2}f_n}\leq1$ and $\twonorm{P_n^{\alpha/2}f_n}\leq1$. The first of the last two inequalities gives:
\begin{align*}
1&\geq\twonorm{V_n^{1/2}f_n}^2\stackrel{(*)}{\geq}\int_{|x|\geq r}v(x)|f_n(x)|^2\,dx\geq\inf_{|x|\geq r}v(x)\int_{|x|\geq r}|f_n(x)|^2\,dx\\
&\Longrightarrow \int_{|x|\geq r}|f_n(x)|^2\,dx\leq\frac{1}{\inf_{|x|\geq r}v(x)}\to0\,,
\end{align*}
uniformly in $n$ as $r\to\infty$. For the inequality $(*)$ we used that $(V_nf_n)(x)=v(x)f_n(x)$ for $|x|\geq r>q^{-1}$.

Next we use the inequality $\twonorm{P_n^{\alpha/2}f_n}\leq1$, valid for all $n$. First we note that 
\begin{align*}
\twonorm{P_n^{\alpha/2}f_n}&=\twonorm{C_nP^{\alpha/2}f_n}=\twonorm{C_n\foi Q^{\alpha/2}\fo f_n}=
\twonorm{\foi S_nQ^{\alpha/2}\fo f_n}\\&=\twonorm{S_nQ^{\alpha/2}\fo f_n}
\end{align*}
which gives, for all $\rho>0$:
\begin{align*}
1&\geq \twonorm{S_nQ^{\alpha/2}\fo f_n}^2\geq\int_{|x|\geq \rho}|(S_nQ^{\alpha/2}\fo f_n)(x)|^2\,dx\stackrel{(*)}{=}
\int_{|x|\geq \rho}|(Q^{\alpha/2}\fo f_n)(x)|^2\,dx
\\&=
\int_{|x|\geq \rho}|x|^{\alpha}\cdot|(\fo f_n)(x)|^2\,dx\geq \rho^\alpha \int_{|x|\geq \rho}|(\fo f_n)(x)|^2\,dx\\
&\Longrightarrow \int_{|x|\geq \rho}|(\fo f_n)(x)|^2\,dx\leq\frac{1}{\rho^\alpha}\to 0
\end{align*}
uniformly in $n$ as $\rho\to\infty$. For the equality $(*)$ we used that $Q^{\alpha/2}\fo f_n$ is locally constant away from the origin. Uniform compactness of the $(I+H_n)^{-1}$ now follows from Proposition~\ref{KR}.
\end{proof}
It now follows that with $M_n=(I+H_n)^{-1}$ and $M=(I+H)^{-1}$ all the conditions of Proposition~\ref{mainprop} are satisfied, and via spectral mapping we can state the analog of Proposition~\ref{mainprop} for $H_n$ and $H$:
\begin{theorem}[cfr.\ Theorem 4 in \cite{DVV94}]\label{mainthm}
\begin{enumerate}
\item If $J$ is a compact subset of $[0,\infty)$ with $J\cap\sigma(H)=\emptyset$, then $J\cap\sigma(H_n)=\emptyset$ for large $n$.
\item If $\gl\in\sigma(H)$, there exists a sequence $(\gl_n)$ with $\gl_n\in\sigma(H_n)$ such that $\gl_n\to\gl$. Further, if $J$ is a compact neighborhood of an eigenvalue $\gl\in\sigma(H)$, not containing any other eigenvalues of $H$, then any sequence $\gl_n$ with $\gl_n\in\sigma(H_n)\cap J$ converges to $\gl$.
\item Let $\gl$ and $J$ be as in (2). Then $\dim P^{H_n}(J) = \dim P^H(J)$ for large $n$, 
and for each orthonormal basis $\{e_1,\dots,e_m\}$ for $r\left(P^{H}(J)\right)$ there is, for each $n$, an orthonormal basis $\{e_1^n,\dots,e^n_m\}$ for $r\left(P^{H_n}(J)\right)$ such that $\lim_{n\to\infty}e^n_i=e_i$, $i=1,\dots,m$.  
\end{enumerate}
\end{theorem}
\section{Numerical investigation of the Schr\"odinger operator over $\q_3[\sqrt{3}]$}
\label{numeric}
\subsection{Overview}
In \cite[Ch.\ 3, Section XII]{VVZ94} a detailed analysis was carried out on the spectrum of the $p$-adic 
Schr\"odinger operator, and in \cite[Ch.\ 3]{MR1848777} a similar analysis was performed on the Schr\"odinger operator over a general local field. 

Let as before $H=\pa +V$ denote the Schr\"odinger operator over a local field $K$. The eigenfunctions of $H$ can be divided into two main types, corresponding to two complementary subspaces of $L^2(K)$: those which are supported on a single spherical shell (which we shall call shell functions), and those which are radial\footnote{With notation as in \cite{VVZ94,MR1848777}, the set of shell functions comprises all the type~I functions plus the shell functions of type~II; the radial functions are all of type~II.}. Of these, only the shell functions are completely understood: They belong to eigenvalues which can be determined from Diophantine equations, and there are explicit formulae for them. For radial eigenfunctions no such explicit formulae seem to be known.

In this numerical study we specialize to the case of the Schr\"odinger operator $H=\frac{1}{2}(P^2+Q^2)$ of the harmonic oscillator over the local field $\q_3[\sqrt{3}]$, which is a quadratic and totally ramified extension of $\q_3$. We were interested in the following questions:
\begin{itemize}
\item Do eigenfunctions of both types (shell functions and radial functions) show up already at the finite level?
\item Is there good agreement between the theoretical and numerical eigenvalues?
\item Is there good agreement between the theoretical and numerical eigenfunctions?
\item Are multiplicities correct?
\end{itemize}
The answer to all these questions was 'yes'. To illustrate this, we sum up some of the results in Table~\ref{eigvals}.
\subsection{More details about the numerical experiment}
The extension $\qtt/\q_3$ is totally ramified, so with notation as in section~\ref{localfields} we have $e=2$, and hence $f=1$ since $ef=[\qtt:\q_3]=2$. Further, from $q=p^f$ follows $q=p=3$, and as uniformizer we can take $\beta=\sqrt{3}$, hence $|\beta|=1/q=1/3$. For the exponent of the different we have $d=1$, so the character $\chi$ defined in subsection~\ref{ft} becomes  
$\chi(x)=\exp\left(2\pi i\{\tr_{\qtt/\q_3}(\sqrt{3}^{-1}x)\}\right)$, $x\in\qtt$.

For the finite model we did experiments with $n=1,2,3,4$, so we were working with finite grids of sizes $|X_1|= 9$, $|X_2|=9^2=81$, $|X_3|=9^3=729$, and $|X_4|=9^4=6561$ï, respectively. Of particular interest to us was how the eigenfunctions came out: Would they clearly exhibit characteristics as shell functions or radial functions? They did. To illustrate this we give in Table~\ref{eigenfunctions} an excerpt from the value tables of three eigenfunctions: one is radial, one is a linear combination of two shell functions, and one is a pure shell function. We also wanted to compare our numerically computed eigenfunctions to the theoretical ones (evaluated on the grid). To do this, we measured the distance from each of the former to the linear span of the latter. Up to machine accuracy ($10^{-16}$), the distance came out as zero. We find this quite remarkable.  

\providecommand{\bysame}{\leavevmode\hbox to3em{\hrulefill}\thinspace}
\providecommand{\MR}{\relax\ifhmode\unskip\space\fi MR }
\providecommand{\MRhref}[2]{%
  \href{http://www.ams.org/mathscinet-getitem?mr=#1}{#2}
}
\providecommand{\href}[2]{#2}

\newpage

\appendix
\section{Tables for numerical eigenvalues and eigenfunctions}
The tables in this section should be self-explanatory\footnote{In the estimate for the lowest eigenvalue in Table~\ref{eigvals} (first entry in column 1) we are assuming that the estimate given in \cite[p. 190]{VVZ94} is valid also in our setting. We haven’t checked this in detail, but there are strong indications that it is true.}. The data are taken from a computer run with $n=2$ (i.e., $81$ points in the finite grid). Each of the functions in Table~\ref{eigenfunctions} is represented with 28 values, with values coming from each of the 5 shells which occur for $n=2$.  

\begin{table}[b]\caption{Numerical approximations to the spectral data of $H=\frac{1}{2}(P^2+Q^2)$ over $\q_3[\sqrt{3}]$.}\label{eigvals}
\begin{tabular}{|r|r|r|r|l|l|}\hline
\parbox[t]{1.6cm}{Theoretical eigenvalue}&\parbox[t]{1.6cm}{Numerical eigenvalue}&\parbox[t]{1.6cm}{Theoretical multiplicity}&\parbox[t]{1.6cm}{Numerical multiplicity}&\parbox[t]{1.6cm}{\raggedright Type of eigenfunction}&Comment\\\hline\hline
\parbox{2.2cm}{\raggedright $0<\lambda_0<9/13$\\$\approx0.6923$} &0.6684&1&1&radial&\\\hline
?&4.6922&?&1&radial&\\\hline
?&4.7158&?&1&radial&\\\hline
5&5.0000&2&2&shell function&\parbox[t]{1.6cm}{\raggedright $2=1+1$: Coming from two different shells.}\\\hline
9&9.0000&4&4&shell function&\parbox[t]{1.6cm}{\raggedright All supported on the same shell.}\\\hline
?&40.5213&?&2&radial&\\\hline
\parbox[t]{1.6cm}{$40+5/9=40.5555\dots$}&40.5555&2&2&shell function&\parbox[t]{1.6cm}{\raggedright $2=1+1$: Coming from two different shells.}\\\hline
41&41.0000&8&8&shell function&\parbox[t]{1.6cm}{\raggedright $8=4+4$: Coming from two different shells.}\\\hline
45&45.0000&24&24&shell function&\parbox[t]{1.6cm}{\raggedright $24=12+12$: Coming from two different shells.}\\\hline
\end{tabular}\vspace{0.5\baselineskip}
\end{table}
\newpage

\begin{table}\caption{Eigenfunctions for three different eigenvalues, 28 values for each function, coming from all the 5 shells. Both kinds of eigenfunctions occur (shell functions and radial functions). -- Shell no.\ $k$ ($k=2,1,0,-1,-\infty$) is the shell $|x|=3^k$ (so shell no.\ $-\infty$ is the shell $|x|=3^{-\infty}=0$).}\label{eigenfunctions}
\begin{tabular}{|r|r|r|r|r|r|}\hline
\multicolumn{2}{|l|}{\parbox{1.7in}{Eigenfunction for the lowest eigenvalue \fbox{$\lambda\approx0.6684$}.\\ It exhibits a perfect radial behavior. Notice also that the function is strictly positive, in accordance with the corresponding statement for the case $K=\qp$ in \cite[p.\ 186]{VVZ94}.}}&\multicolumn{2}{|l|}{\parbox{1.7in}{Eigenfunction for \fbox{$\lambda=5$}.\\ Eigenfunctions here are linear combinations of shell functions from two different shells (shells 1 and 0). As should be expected, the function below  exhibits non-radial behavior, being non-constant on each shell where it doesn't vanish (shells 1 and 0).}}&\multicolumn{2}{|l|}{\parbox{1.7in}{Eigenfunction for 
\fbox{$\lambda=9$}. It exhibits a perfect shell function behavior, with support on shell no.\ 1.}}\\\hline\hline
&Shell no.&&Shell no.&&Shell no.\\\hline
$3.5818432\cdot10^{-1}$&$-\infty$&$1.8757870\cdot10^{-15}\approx0$&$-\infty$&$-3.8765003\cdot10^{-16}\approx0$&$-\infty$\\\hline
$5.5430722\cdot10^{-5}$&2&$2.0896995\cdot10^{-16}\approx0$&2&$1.6021680\cdot10^{-16}\approx0$&2\\
$5.5430722\cdot10^{-5}$&2&$8.7737711\cdot10^{-17}\approx0$&2&$-9.1411700\cdot10^{-17}\approx0$&2\\
$5.5430722\cdot10^{-5}$&2&$-1.4801152\cdot10^{-16}\approx0$&2&$5.1268297\cdot10^{-17}\approx0$&2\\
$5.5430722\cdot10^{-5}$&2&$3.0773313\cdot10^{-16}\approx0$&2&$2.7677667\cdot10^{-16}\approx0$&2\\
$5.5430722\cdot10^{-5}$&2&$-4.5409159\cdot10^{-17}\approx0$&2&$-4.5822760\cdot10^{-16}\approx0$&2\\
$5.5430722\cdot10^{-5}$&2&$-1.0479409\cdot10^{-16}\approx0$&2&$-1.3758518\cdot10^{-16}\approx0$&2\\
$5.5430722\cdot10^{-5}$&2&$-2.3471948\cdot10^{-17}\approx0$&2&$2.1385872\cdot10^{-17}\approx0$&2\\
$5.5430722\cdot10^{-5}$&2&$7.9466194\cdot10^{-17}\approx0$&2&$-1.0549816\cdot10^{-16}\approx0$&2\\
$5.5430722\cdot10^{-5}$&2&$2.3950293\cdot10^{-16}\approx0$&2&$2.3917324\cdot10^{-16}\approx0$&2\\
$5.5430722\cdot10^{-5}$&2&$6.4773691\cdot10^{-17}\approx0$&2&$1.2912546\cdot10^{-16}\approx0$&2\\
$5.5430722\cdot10^{-5}$&2&$-1.1431061\cdot10^{-16}\approx0$&2&$-6.0210598\cdot10^{-17}\approx0$&2\\
$5.5430722\cdot10^{-5}$&2&$-1.3177515\cdot10^{-17}\approx0$&2&$-3.9251100\cdot10^{-17}\approx0$&2\\
$5.5430722\cdot10^{-5}$&2&$1.3595786\cdot10^{-16}\approx0$&2&$-5.0103544\cdot10^{-17}\approx0$&2\\
$5.5430722\cdot10^{-5}$&2&$3.2839452\cdot10^{-17}\approx0$&2&$1.2137971\cdot10^{-16}\approx0$&2\\
$5.5430722\cdot10^{-5}$&2&$7.8206625\cdot10^{-17}\approx0$&2&$-1.0063910\cdot10^{-16}\approx0$&2\\
$5.5430722\cdot10^{-5}$&2&$3.3933100\cdot10^{-17}\approx0$&2&$-7.7900493\cdot10^{-17}\approx0$&2\\
$5.5430722\cdot10^{-5}$&2&$8.8459742\cdot10^{-17}\approx0$&2&$2.2672330\cdot10^{-16}\approx0$&2\\
$5.5430722\cdot10^{-5}$&2&$2.2115193\cdot10^{-17}\approx0$&2&$-1.1819127\cdot10^{-16}\approx0$&2\\\hline
$1.2747433\cdot10^{-2}$&1&$-2.3459638\cdot10^{-1}$&1&$5.9907185\cdot10^{-2}$&1\\
$1.2747433\cdot10^{-2}$&1&$2.3459638\cdot10^{-1}$&1&$-4.1084268\cdot10^{-1}$&1\\
$1.2747433\cdot10^{-2}$&1&$-2.3459638\cdot10^{-1}$&1&$-1.0595734\cdot10^{-1}$&1\\
$1.2747433\cdot10^{-2}$&1&$2.3459638\cdot10^{-1}$&1&$2.7644342\cdot10^{-2}$&1\\
$1.2747433\cdot10^{-2}$&1&$-2.3459638\cdot10^{-1}$&1&$4.6050157\cdot10^{-2}$&1\\
$1.2747433\cdot10^{-2}$&1&$2.3459638\cdot10^{-1}$&1&$3.8319834\cdot10^{-1}$&1\\\hline
$3.1960943\cdot10^{-1}$&0&$3.9500330\cdot10^{-2}$&0&$1.2637350\cdot10^{-17}\approx0$&0\\
$3.1960943\cdot10^{-1}$&0&$-3.9500330\cdot10^{-2}$&0&$-1.6035100\cdot10^{-17}\approx0$&0\\\hline
$3.5768544\cdot10^{-1}$&-1&$2.2996138\cdot10^{-17}\approx0$&-1&$-9.9411507\cdot10^{-17}\approx0$&-1\\\hline
\end{tabular}

\end{table}

\end{document}